\documentclass[a4paper,11pt]{article}
\usepackage[top=2.5cm, bottom=2.5cm, left=2.5cm, right=2.5cm]{geometry}

\usepackage{amsmath,amssymb,graphicx,pst-all,amscd,amsthm,comment}
\usepackage[utf8x]{inputenc}
\usepackage{tikz}
\usetikzlibrary{matrix,arrows,decorations.pathmorphing}

\newtheorem{theo}{Theorem}
\newtheorem{cor}{Corollary}

\theoremstyle{definition}
\newtheorem{defin}{Definition}

\theoremstyle{remark}

     {\begin{list}{}%
             {\setlength{\leftmargin}{#1}}%
             \item[]%
     }
     {\end{list}}

\newcommand{\tc}{TestCover($p$)}
\newcommand{\tcp}{TestCover($p,k$)}
\newcommand{\tcpk}{TestCover($k,k$)}
\newcommand{\tcpr}{Test-$r$-Cover($p,k$)}
\newcommand{\tcprk}{Test-$r$-Cover($k,k$)}
\newcommand{\tcpn}{TestCover($n-k,k$)}

\begin{document}

\title{(Non-)existence of Polynomial Kernels for the Test Cover Problem}

\author{
 G. Gutin, G. Muciaccia\\Royal Holloway, University of London, Egham TW20 0EX, UK \and A. Yeo\\ University of Johannesburg,
Auckland Park, 2006 South Africa
}
\date{}
\maketitle

\begin{abstract}
The input of the Test Cover problem consists of a set $V$ of vertices,
and a collection ${\cal E}=\{E_1,\ldots , E_m\}$ of distinct subsets of $V$, called tests. A test
$E_q$ separates a pair $v_i,v_j$ of vertices if $|\{v_i,v_j\}\cap E_q|=1.$ A subcollection ${\cal T}\subseteq {\cal E}$
is a test cover if each pair $v_i,v_j$ of distinct vertices is separated by a test in ${\cal T}$.
The objective is to find a test cover of minimum cardinality, if one exists. This problem is NP-hard.

We consider two parameterizations the Test Cover problem with parameter $k$: (a) decide whether there is a test cover with at most $k$ tests, (b) decide whether there is a test cover with at most $|V|-k$ tests. Both parameterizations are known to be fixed-parameter tractable. We prove that none have a polynomial size kernel unless $NP\subseteq coNP/poly$. Our proofs use the cross-composition method recently introduced by Bodlaender et al. (2011) and parametric duality introduced by Chen et al. (2005). The result for the parameterization (a) was an open problem (private communications with Henning Fernau and Jiong Guo, Jan.-Feb. 2012). We also show that the parameterization (a) admits a polynomial size kernel if the size of each test is upper-bounded by a constant.
\end{abstract}

\section{Introduction}
In the Test Cover problem defined below we are given a set $V=\{v_1,\ldots, v_n\}$ of element called {\em vertices} and
a family  $\mathcal{E}=\{E_1,\dots,E_m\}$ of distinct subsets of $V$ called {\em tests}. We say that a test
$E_q$ {\em separates} a pair $v_i,v_j$ of vertices if $|\{v_i,v_j\}\cap E_q|=1.$
A subset $\mathcal{T}$ of $\mathcal{E}$ is called a {\em test cover}
if for every pair of distinct vertices $v_i,v_j$ there exists a test $T\in\mathcal{T}$ separating them.

\begin{center}
\fbox{~\begin{minipage}{0.9\textwidth}
{\sc TestCover}($p$)

\textbf{Input:} A set  $V=\{v_1,\ldots, v_n\}$ of vertices, a family  $\mathcal{E}=\{E_1,\dots,E_m\}$ of tests and an integer $p$.

\textbf{Task:} Decide whether there exists a test cover with at most $p$ tests.
\end{minipage}~}
\end{center}

The (optimization version of) Test Cover problem arises naturally in the following general setting of
identification problems: Given a set of items (i.e., vertices) and a set
of binary attributes that may or may not occur in each item, the aim
is to find a minimum-size subset of attributes (a minimum test cover) such that each
items can be uniquely identified from the information on the subset
of attributes it contains.
The Test Cover problem arises in fault analysis, medical diagnostics, pattern recognition, and
biological  identification (see, e.g., \cite{HalHalRav01,HalMinRav01,MorSha85}).

The {\sc TestCover}($p$) problem is NP-hard, as was shown by Garey and Johnson \cite{GarJoh79}. Moreover,
{\sc TestCover}($p$) is APX-hard \cite{HalHalRav01}. There is an $O(\log n)$-approximation algorithm for the problem \cite{MorSha85}
and there is no $o(\log n)$-approximation algorithm unless P=NP \cite{HalHalRav01}.

The generic parameterized version \tcp\space of {\sc TestCover}($p$) is defined as follows.

\begin{center}
\fbox{~\begin{minipage}{0.9\textwidth}
{\sc TestCover}$(p,k)$

\textbf{Input:} A set $V=\{v_1,\dots,v_n\}$, a family $\mathcal{E}=\{E_1,\dots,E_m\}$ of subsets of $V$ and two integers $p$ and $k$.

\textbf{Parameter:} $k$.

\textbf{Task:} Decide whether there exists a test cover with at most $p$ tests.
\end{minipage}~}
\end{center}

In this paper, we study two parameterizations of {\sc TestCover}: the standard parameterization {\sc TestCover}$(k,k)$ and its dual\footnote{See Definition \ref{def:dual} for a formal notion of parameterized duality.} {\sc TestCover}$(n-k,k)$.
Both parameterizations are known to be fixed-parameter tractable \cite{CroGutJonSauYeo}.
Fixed-parameter tractability of {\sc TestCover}$(k,k)$ is easy to see using the fact that there is no test cover of size less than $\lceil \log n\rceil$ \cite{HalHalRav01}. Indeed, if $k< \log n$ then the corresponding instance of {\sc TestCover}$(k,k)$ is a {\sc No}-instance, and if $k\ge \log n$ then $n\le 2^k$ and so $m\le 2^{2^k}$ allowing us to solve {\sc TestCover}$(k,k)$ by a brute force fixed-parameter algorithm.
The proof in \cite{CroGutJonSauYeo} that {\sc TestCover}$(n-k,k)$ is fixed-parameter tractable is much harder.

We first prove that {\sc TestCover}$(k,k)$ does not admit a polynomial kernel unless $NP\subseteq coNP/poly$. The latter is deemed unlikely as it is known to imply a collapse of the polynomial hierarchy to its third level \cite{Yap}. Note that our first result solves an open problem \cite{Fer,Guo}. Our proof uses the cross-composition method recently introduced by Bodlaender et al. \cite{BodJanKra}.

We show that {\sc TestCover}$(k,k)$ does admit a polynomial kernel if the size of each test is bounded by a constant. This constraint is used in some practical applications of {\sc TestCover} \cite{HalHalRav01}.  We prove that {\sc TestCover}$(n-k,k)$ does not admit a polynomial kernel unless $NP\subseteq coNP/poly$. This result follows from our first result and a general result on nonexistence of polynomial kernels in dual parameterized problems also obtained in our paper.
We use the notion of dual introduced by Chen et al. \cite{CheFerKanXia} as well as the cross-composition method.

Our paper is organized as follows. In the rest of this section we will give a couple of simple, yet very useful, definitions on {\sc TestCover}.
In the next section we give basics on parameterized algorithms and kernelization as well as on the cross-composition method. In Section \ref{sec:res}, we prove all our results. In Section \ref{sec:op} we pose some open problems.

A {\em partial test cover} $\mathcal{T}'$ to an instance of \tcp\space is a subset of $\mathcal{E}$ of cardinality at most $p$.
We say that $C\subseteq V$ is a {\em class} induced by a partial test cover $\mathcal{T}'$, if $C$ is a maximal set such that there does not exist
a test in $\mathcal{T}'$ which separates two vertices of $C$. Notice that the classes induced by every partial test cover form a partition of $V$ and every test cover induces $n$ classes. For a positive integer $p$, let $[p]=\{1,2,\ldots ,p\}.$

\section{Fixed-Parameter Algorithms, Kernels and Cross-composition Method}

A \emph{parameterized problem} $P$ is a subset $P\subseteq \Sigma^* \times
\mathbb{N}$ over a finite alphabet $\Sigma$. $P$ is
\emph{fixed-parameter tractable} if the membership of an instance
$(x,k)$ of $\Sigma^* \times \mathbb{N}$ in $P$ can be decided by an algorithm of runtime
$f(k)|x|^{O(1)}$, where $f$ is a function of the
{\em parameter} $k$ only (such an algorithm is a \emph{fixed-parameter} algorithm)
\cite{DowneyFellows99,FlumGrohe06,Niedermeier06}.
Given a parameterized problem $P$,
a \emph{kernelization of $P$} is a polynomial-time
algorithm that maps an instance $(x,k)$ to an instance $(x',k')$ (the
\emph{kernel}) such that (i)~$(x,k)\in P$ if and only if
$(x',k')\in P$, (ii)~ $k'\leq h(k)$, and (iii)~$|x'|\leq g(k)$ for some
functions $h$ and $g$. The function $g(k)$ is called the {\em size} of the kernel.
It is well-known \cite{DowneyFellows99,FlumGrohe06,Niedermeier06} that a decidable parameterized problem $P$ is fixed-parameter
tractable if and only if it has a kernel. Polynomial-size kernels are of
main interest, due to applications \cite{DowneyFellows99,FlumGrohe06,Niedermeier06}, but unfortunately many fixed-parameter tractable problems
have no such kernels unless  coNP$\subseteq$NP/poly, see, e.g., \cite{BDFH09}.\\

The following two definitions and Theorem \ref{res} were given by Bodlaender \textit{et al.} \cite{BodJanKra}.

\begin{defin}[Polynomial equivalence relation]\label{polequ}
An equivalence relation $\mathcal{R}$ on $\Sigma^\ast$ is called a {\em polynomial equivalence relation} if the following two conditions hold:
\begin{itemize}
   \item{There is an algorithm that given two strings $x,y\in\Sigma^\ast$ decides whether $x$ and $y$ belong to the same equivalence class
in $(|x|+|y|)^{O(1)}$ time.}
    \item{For any finite set $S\subseteq\Sigma^\ast$ the equivalence relation $\mathcal{R}$ partitions the elements of $S$ into at most
$(\max_{x\in S}|x|)^{O(1)}$ equivalence classes.}
\end{itemize}
\end{defin}

\begin{defin}[Cross-composition]\label{crocom}
Let $L\subseteq\Sigma^\ast$ be a problem and let $Q\subseteq\Sigma^\ast\times\mathbb{N}$ be a parameterized problem. We say that $L$
{\em cross-composes} into $Q$ if there is a polynomial equivalence relation $\mathcal{R}$ and an algorithm which, given $t$ strings $x_1,\dots,x_t$
belonging to the same equivalence class of $\mathcal{R}$, computes an instance $(x^\ast,k^\ast)\in\Sigma^\ast\times\mathbb{N}$ in time
polynomial in $\sum_{i=1}^t|x_i|$ such that:
\begin{itemize}
   \item {$(x^\ast,k^\ast)\in Q$ if and only if $x_i\in L$ for some $1\leq i\leq t$.}
  \item{$k^\ast$ is bounded by a polynomial in $\max_{i=1}^t|x_i|+\log t$.}
\end{itemize}

\end{defin}

\begin{theo}\label{res}
If some problem $L$ is $NP$-hard under Karp reductions and $L$ cross-composes into the
parameterized problem $Q$ then there is no polynomial kernel for $Q$ unless $NP\subseteq coNP/poly$.
\end{theo}

\section{Results}\label{sec:res}

\begin{theo}\label{nopolker}
\tcpk\space does not admit a polynomial kernel, unless $NP\subseteq coNP/poly$.
\end{theo}
\begin{proof}
We will use Theorem \ref{res} to show the result, hence we need an $NP$-hard problem $L$ which cross-composes into \tcpk.
For this purpose, it is possible to use \tc.
An instance of this problem is a triple $(V,\mathcal{E},p)$. We say that two triples $(V_1,\mathcal{E}_1,p_1)$
and $(V_2,\mathcal{E}_2,p_2)$ are equivalent if $|V_1|=|V_2|$ and $p_1=p_2$.

It is not difficult to see that this defines a polynomial
equivalence relation on \tc\space (see Definition \ref{polequ}). To see that the second condition of Definition \ref{polequ} holds, observe that the relation
thus defined partitions a finite set $S$ into at most $O(n\cdot m)$ equivalence classes, assuming that for every $(V,\mathcal{E},p)\in S$ we have
that $|V|\leq n$ and $|\mathcal{E}|\leq m$ (since $p\leq|\mathcal{E}|$).

Consider now $t$ instances $Q_1,\dots,Q_t$ of \tc, belonging to the same equivalence class. For every $i\in[t]$,
let $V=\{v_1,\dots,v_n\}$ be the set of vertices of $Q_i$, $\mathcal{E}_i=\{E_1^i,\dots,E_{m_i}^i\}$ its set of tests and
$p$ be the upper bound on the size of the solution.
Let $l=2\lceil(\log t)/2\rceil$.
We will construct an instance $Q^\ast$ of \tcpk\space such that $Q^\ast$ has a solution of size at most $k=2l+p$
if and only if at least one of $Q_1,\dots,Q_t$ has a solution of size at most $p$.
This will be enough to prove our result (see Definition \ref{crocom} and Theorem \ref{res}). \\

The vertex set of $Q^\ast$ is defined in the following way:
$$V(Q^\ast)=\{v_1,\dots,v_n\}\cup(\cup_{j=1}^{2l}\{y_j,x_1^j,\dots,x_p^j\})\cup\{a_1,\dots,a_l\}.
$$
The tests in $Q^\ast$ are defined as follows. Consider some $i\in[t]$ and let $(i_1,\dots,i_l)$ be its binary representation
(e.g. if $i=11$ and $l=6$, $(i_1,\dots,i_l)=(1,1,0,1,0,0)$). For every $i$, consider the following family of sets $\mathcal{S}_i$
(where subscripts are taken modulo $k$):
$$\mathcal{S}_i=\bigl\{\{x_{h}^1,x_{h+i_1}^2,x_{h}^3,x_{h+i_2}^4,x_{h}^5,x_{h+i_3}^6,\dots,x_{h}^{2l-1},x_{h+i_l}^{2l}\}:
h=1,\dots,p\bigr\}.
$$
Define $\mathcal{E}(Q^\ast)=\mathcal{E}^\ast\cup\widetilde{\mathcal{E}}_1\cup\dots\cup\widetilde{\mathcal{E}}_t$, where
$$\mathcal{E}^\ast=\bigl\{\{a_j,y_{2j-1},x_1^{2j-1},x_2^{2j-1},x_3^{2j-1},\dots,x_p^{2j-1}\},\{a_j,y_{2j},x_1^{2j},x_2^{2j},x_3^{2j}
\dots,x_p^{2j}\}:j\in[l]\bigr\}
$$
and
$$\widetilde{\mathcal{E}}_i=\{(S\cup E_j^i):j\in[m_i],S\in\mathcal{S}_i\}.
$$
This completes the definition of $Q^\ast$.

Assume that there is a solution to $Q^\ast$ using at most $2l+p$ tests. Since $a_i$ and $y_{2i-1}$ need to be separated and
similarly for $a_i$ and $y_{2i}$ for all $i$, all $2l$ tests in $\mathcal{E}^\ast$ must be used (because they are the only tests
where $a_i$, $y_{2i-1}$ and $y_{2i}$ appear); notice also that they are enough to separate them.
Moreover, as every set in $\widetilde{\mathcal{E}}_i$ (for all $i$) only intersects $\mathcal{X}_i=\{x_1^i,x_2^i,\dots,x_p^i\}$
in one vertex and as all vertices in $\mathcal{X}_i$ have to be separated from $y_i$, the tests not in $\mathcal{E}^\ast$ all have to contain
a distinct vertex from $\mathcal{X}_i$
\footnote{Notice that this ensures that a solution to $Q^\ast$ has \textit{at least} $2l+p$ vertices. Therefore we are actually showing that
$Q^\ast$ has a solution of size exactly $4\lceil(\log t)/2\rceil+p$
if and only if at least one of $Q_1,\dots,Q_t$ has a solution of size at most $p$.}.

Consider now a test $E$ in $\widetilde{\mathcal{E}}=\bigcup_{j=1}^t\widetilde{\mathcal{E}}_j$. Its intersection with $\mathcal{X}_i
\cup\mathcal{X}_{i+1}$, for every $i\in\{1,3,\dots,2l-1\}$, contains only vertices which are in $S$ for some $S\in\bigcup_{j=1}^t\mathcal{S}_j$.

Suppose $S\in\mathcal{S}_r$ for some $r$ (this ensures that $E\in\widetilde{\mathcal{E}}_r$)
and suppose that it is equal to
$$\{x_{h}^1,x_{h+r_1}^2,\dots,x_{h}^{2l-1},x_{h+r_l}^{2l}\},$$
where $(r_1,\dots,r_l)$ is the binary representation of $r$ and $h\in[p]$.

The intersection with $\mathcal{X}_i\cup\mathcal{X}_{i+1}$ is given by the black
vertices in next figure if $r_{\frac{i+1}{2}}=0$:

\begin{equation*}
\begin{tikzpicture}

\node (pun1) at (3,0) {$\dots$};
\node (pun2) at (6,0) {$\dots$};

\node (pun3) at (3,-1) {$\dots$};
\node (pun4) at (6,-1) {$\dots$};

\tikzstyle{every node}=[draw,circle,fill=white,minimum size=4pt,inner sep=0pt]

\foreach \number in {1,2}
{
\node (N1 + \number) at (\number,0) [label=above:$x_\number^i$] {};
}

\node[fill=black] (N1quattro) at (4,0) [label=above:$x_{h}^i$] {};
\node (N1cinque) at (5,0) [label=above:$x_{h+1}^i$] {};

\node (N1sette) at (7,0) [label=above:$x_{p-1}^i$] {};
\node (N1otto) at (8,0) [label=above:$x_p^i$] {};

\foreach \number in {1,2}
{
\node (N2 + \number) at (\number,-1) [label=below:$x_\number^{i+1}$] {};
}

\node[fill=black] (N2quattro) at (4,-1) [label=below:$x_{h}^{i+1}$] {};
\node (N2cinque) at (5,-1) [label=below:$x_{h+1}^{i+1}$] {};

\node (N2sette) at (7,-1) [label=below:$x_{p-1}^{i+1}$] {};
\node (N2otto) at (8,-1) [label=below:$x_p^{i+1}$] {};

\end{tikzpicture}
\end{equation*}

If, on the contrary, $r_{\frac{i+1}{2}}=1$, the intersection is given by the black vertices in next figure:

\begin{equation*}
\begin{tikzpicture}

\node (pun1) at (3,0) {$\dots$};
\node (pun2) at (6,0) {$\dots$};

\node (pun3) at (3,-1) {$\dots$};
\node (pun4) at (6,-1) {$\dots$};

\tikzstyle{every node}=[draw,circle,fill=white,minimum size=4pt,inner sep=0pt]

\foreach \number in {1,2}
{
\node (N1 + \number) at (\number,0) [label=above:$x_\number^i$] {};
}

\node[fill=black] (N1quattro) at (4,0) [label=above:$x_{h}^i$] {};
\node (N1cinque) at (5,0) [label=above:$x_{h+1}^i$] {};

\node (N1sette) at (7,0) [label=above:$x_{p-1}^i$] {};
\node (N1otto) at (8,0) [label=above:$x_p^i$] {};

\foreach \number in {1,2}
{
\node (N2 + \number) at (\number,-1) [label=below:$x_\number^{i+1}$] {};
}

\node (N2quattro) at (4,-1) [label=below:$x_{h}^{i+1}$] {};
\node[fill=black] (N2cinque) at (5,-1) [label=below:$x_{h+1}^{i+1}$] {};

\node (N2sette) at (7,-1) [label=below:$x_{p-1}^{i+1}$] {};
\node (N2otto) at (8,-1) [label=below:$x_p^{i+1}$] {};

\end{tikzpicture}
\end{equation*}

These are the only possible types of intersection.

Notice that, if $E$ is part of the solution, in order to cover all vertices in $\mathcal{X}_i\cup\mathcal{X}_{i+1}$ (for every $i$)
it is possible to use only tests in $\widetilde{\mathcal{E}}_r$. Otherwise, if a test in $\widetilde{\mathcal{E}}_s$
is used (with $r\neq s$),
there exists $i^\ast$ such that $r_{i^\ast} \neq s_{i^\ast}$ and it is impossible to cover all vertices
in $\mathcal{X}_{i^\ast}\cup\mathcal{X}_{i^\ast+1}$.
Therefore, if there is a solution to $Q^\ast$ using $2l+p$ tests, then $2l$ tests are taken from $\mathcal{E}^\ast$ and $p$ tests from
$\widetilde{E}_r$ for some $r$. This ensures that $Q_r$ has a solution using at most $p$ tests. \\

To prove the opposite direction, assume that $Q_r$ has a solution of size at most $p$. Using all sets in $\mathcal{E}^\ast$ and carefully choosing
$p$ sets from $\widetilde{\mathcal{E}}_r$
\footnote{The $p$ tests from $\widetilde{\mathcal{E}}_r$ have to contain the tests used in a solution to $Q_r$ and tests $S$ which are chosen
accordingly to the binary representation of $r$, in order to cover all vertices in $\bigcup_{j=1}^{2l}{\mathcal{X}_j}$.},
it is possible to obtain a solution to $Q^\ast$ with at most $2l+p$ tests.
This completes the proof.
\end{proof}

\begin{defin}
The generic parameterized problem \tcpr\space is defined as \tcp, except that for an instance $(V,\mathcal{E},p,k)$
we have that $r\leq|V|$ and $|E|\leq r$ for every $E\in\mathcal{E}$.
\end{defin}

As the next theorem shows, this additional condition enables us to obtain a kernel which is linear in the number of vertices
in the case of \tcprk.

\begin{theo}\label{polver}
\tcprk\space admits a kernel with at most $r\cdot k-(\lfloor\log r\rfloor-1)\cdot r$ vertices.
\end{theo}
\begin{proof}
Notice that $s$ tests can create at most $2^s$ classes: more specifically, adding a test to a partial test cover which induces $t$ classes
produces a partial test cover with at most $2t$ classes. This happens because in order to produce a new class, the test must contain a subset
$C'$ of a class $C$ such that $C'$ is not empty and it does not coincide with $C$ and if this happens for every class, the number of classes doubles.
At the same time, since every test contains at most $r$ vertices, this can happen with at most $r$ different classes.
Therefore, given a partial test cover which induces $t$ classes, adding a new test can create at most $\min\{t,r\}$ new classes.
Using $k$ tests it is possible to create at most $2^{\lfloor\log r\rfloor}+(k-\lfloor\log r\rfloor)r$ classes,
since the first $\lfloor\log r\rfloor$ tests can create at most $2^{\lfloor\log r\rfloor}$ classes and
every other test increases the number of classes by at most $r$.
It follows that we have a {\sc No}-instance unless
$$n\leq 2^{\lfloor\log r\rfloor}+(k-\lfloor\log r\rfloor)r\leq r+(k-\lfloor\log r\rfloor)r=kr-(\lfloor\log r\rfloor -1)r.
$$
\end{proof}

\begin{cor}\label{cor1}
\tcprk\space admits a polynomial kernel when $r$ is a constant.
\end{cor}
\begin{proof}
Theorem \ref{polver} ensures the existence of a kernel with at most $r\cdot k-(\lfloor\log r\rfloor-1)\cdot r$ vertices. This means that the number
$m$ of tests is at most $$\sum_{s=1}^r\binom{rk}{s}\leq\sum_{s=1}^r\left(\frac{erk}{s}\right)^s\in O(r(erk)^r)=O(k^r).$$
\end{proof}

\begin{cor}\label{cor2}
Consider an instance $(V,\mathcal{E},k)$ of \tcpk\space and let $r=\max_{E\in\mathcal{E}}|E|$. If $r\in O(n^{1-\varepsilon})$ for some
$\varepsilon>0$, then the problem admits a kernel which is polynomial in the number of vertices.
\end{cor}
\begin{proof}
By Theorem \ref{polver}, $n<kr$. Now, for some positive constant $c$ it holds that $kr<kcn^{1-\varepsilon}$, which ensures that
$n^{\varepsilon}<kc$, that is $n<(kc)^{\frac{1}{\varepsilon}}$.
\end{proof}

The proof of Theorem \ref{nopolker} can be modified to show the non-existence of a polynomial kernel for the parameterized problem \tcpn.
However, it is possible to obtain this result from a more general theorem. First of all, consider a definition of dual problem equivalent to the one
given by Chen \textit{et al.} \cite{CheFerKanXia}:

\begin{defin}\label{def:dual}
Let $P$ be a parameterized problem, i.e., $P\subseteq \Sigma^* \times
\mathbb{N}$. A mapping $s:\Sigma^*\rightarrow\mathbb{N}$ is a {\em size function} for $P$ if
\begin{itemize}
\item{$0\leq k\leq s(x)$ for every $(x,k)\in P$, and}
\item{$s(x)\leq |x|$ for every $x\in  \Sigma^*$.}
\end{itemize}
The {\em dual} $P_d$ of a problem $P$ respectively to the size function $s$ is the problem corresponding
to the language (i.e., the set of {\sc Yes}-instances)
$P_d=\{(x,s(x)-k):\ (x,k)\in P\}$.
$P_d$ is also called the $s$-{\em dual} of $P$.
\end{defin}

Notice that the $s$-dual of the $s$-dual of a problem is again the original problem.

\begin{theo}\label{dualnopolker}
Given a parameterized problem $P$, if $P$ admits a cross-composition from a NP-hard problem $L$, let $(x,k)$ be the instance
which is associated to the $t$ strings $x_1,\dots,x_t$ by the cross-composition algorithm. If there exists a size function
$s(x)$ which is bounded by a polynomial in
$\max_{i=1}^t|x_i|+\log t$,
then the $s$-dual problem $P_d$ does not
admit a polynomial kernel, unless $NP\subseteq coNP/poly$.
\end{theo}
\begin{proof}
We will show that if the hypothesis hold, $P_d$ admits a cross-composition from the same problem $L$. Hence we conclude using Theorem \ref{res}.
By Definition \ref{crocom}, we know that there exists $i$ such that $x_i\in L$ if and only if $(x,k)\in P$.
By definition of $s$-dual problem, this is equivalent to $(x,s(x)-k)\in P_d$.
This provides a cross-composition for $P_d$, as long as the parameter $k'=s(x)-k$ is a polynomial in $\max_{i=1}^t|x_i|+\log t$,
that is as long as $s(x)$ is a polynomial in $\max_{i=1}^t|x_i|+\log t$.
\end{proof}

Now, notice that \tcpn\space is the $s$-dual problem of \tcpk, with $s(Q)=n$. Notice also that in Theorem \ref{nopolker},
$s(Q^*)=n+2l\cdot(k+1)+l$, which is a polynomial in $n+l\leq\max_{i=1}^t|Q_i|+\log t$. Hence, the hypothesis of Theorem \ref{dualnopolker}
holds and there is no polynomial kernel for \tcpn. Thus, we have the following:

\begin{theo}
The problem \tcpn{}
does not admit a polynomial kernel, unless $NP\subseteq coNP/poly$.
\end{theo}

\section{Open Problems}\label{sec:op}

In Corollaries \ref{cor1} and \ref{cor2} we give  sufficient conditions for the existence of kernels of polynomial size and polynomial in the number of vertices, respectively. It would be interesting to obtain extensions of the corollaries.


\begin{thebibliography}{1}
{\small

\bibitem{BDFH09} H.L.~Bodlaender, R.G.~Downey, M.R~Fellows, and D.~Hermelin, On problems without polynomial kernels,
{\em J. Computer System Sci.} 75(2009), 423--434.

\bibitem{BodJanKra} H.L. ~Bodlaender, B.M.P. ~Jansen and S. ~Kratsch, Cross-Composition: A New Technique for Kernelization Lower Bounds.
{\em Proc. STACS 2011}, Leibniz-Zentrum fuer Informatik 9 (2011), 165--176.

\bibitem{DowneyFellows99} R.~G. Downey and M.~R. Fellows. Parameterized Complexity. Springer, 1999.


\bibitem{CheFerKanXia} J. ~Chen, H. ~Fernau, I. A. ~Kanj and G. ~Xia,
Parametric Duality and Kernelization: Lower Bounds and Upper Bounds on Kernel Size,
Lect. Notes Comput. Sci. 3404 (2005), 269--280.

\bibitem{CroGutJonSauYeo} R. Crowston, G. Gutin, M. Jones,  S. Saurabh and  A. Yeo,
Four Parameterizations of Test Cover Problem. Manuscript.

\bibitem{Fer} Henning Fernau, Private communication, Jan. 2012.

\bibitem{FlumGrohe06}
J.~Flum and M.~Grohe, Parameterized Complexity Theory, Springer Verlag, 2006.

\bibitem{GarJoh79} M.R. Garey and D.S. Johnson, Computers and Intractability:
A Guide to the Theory of NP-Completeness, Freeman and Co., 1979.


\bibitem{Guo} Jiong Guo, Private communication, Feb. 2012.

\bibitem{HalHalRav01} B.V. Halld\'orsson, M.M. Halld\'orsson, and R. Ravi,
On the approximability of the Minimum Test Collection problem. {\em Proc. ESA 2001},
Lect. Notes Comput. Sci. 2161 (2001), 158--169.

\bibitem{HalMinRav01} B.V. Halld\'orsson, J.S. Minden, and R. Ravi. PIER: Protein identification by epitope recognition.
{\em Proc. Currents
in Computational Molecular Biology 2001}, 109--110, 2001.


\bibitem{MorSha85} B.M.E. Moret and H.D. Shapiro, On minimizing a set of tests.
{\em SIAM J. Scientific \& Statistical Comput.} 6 (1985), 983--1003.

\bibitem{Niedermeier06}
R.~Niedermeier. Invitation to Fixed-Parameter Algorithms,
Oxford University Press, 2006.

\bibitem{Yap} C.-K. Yap, Some consequences of non-uniform conditions on uniform classes.
{\em Theor. Comput. Sci.} 26 (1983), 287--300.
}

\end{thebibliography}
\end{document}